\documentclass[10pt]{article}

\usepackage{a4wide}
\usepackage{latexsym}
\usepackage{amssymb}
\usepackage{amsmath}
\usepackage{amsthm}
\usepackage{xspace}
\usepackage[utf8]{inputenc}
\usepackage{enumerate}
\usepackage{authblk}
\usepackage{hyperref}

\DeclareMathOperator{\uclos}{\uparrow}

\newcommand{\canoc}{\ensuremath{\leqslant_\Cs}\xspace}

\newcommand{\pol}[1]{\ensuremath{\mathit{Pol}(#1)}\xspace}

\newcommand{\copol}[1]{\ensuremath{\mathit{co\textup{-}\!Pol}(#1)}\xspace}
\newcommand{\capol}[1]{\ensuremath{\pol{#1} \cap \copol{#1}}\xspace}

\newcommand{\Cs}{\ensuremath{\mathcal{C}}\xspace}
\newcommand{\Ds}{\ensuremath{\mathcal{D}}\xspace}

\newcommand{\Fs}{\ensuremath{\mathcal{F}}\xspace}

\newcommand{\vari}{quotienting Boolean algebra\xspace}

\newcommand{\varis}{quotienting Boolean algebras\xspace}

\newcommand{\pvari}{quotienting lattice\xspace}

\newcommand{\pvaris}{quotienting lattices\xspace}

\def\inv{^{-1}}

\newcounter{sauvegarde}
\newcommand\nat{\ensuremath{\mathbb{N}}\xspace}

\newtheorem{theorem}{Theorem}
\newtheorem{corollary}[theorem]{Corollary}
\newtheorem{proposition}[theorem]{Proposition}

\newtheorem{lemma}[theorem]{Lemma}

\newtheorem{remark}[theorem]{Remark}
\newtheorem{fact}[theorem]{Fact}

\begin{document}

\title{A generic characterization of \pol{\Cs}}
\author[1]{Thomas Place}
\author[2]{Marc Zeitoun}
\affil[1]{Bordeaux University, Labri}
\affil[2]{Bordeaux University, Labri}

\maketitle

\section{Introduction}

We investigate the polynomial closure operation ($\Cs \mapsto \pol{\Cs}$) defined on classes of regular languages. We present an interesting and useful connection relating the separation problem for the class \Cs and the membership problem for it polynomial closure \pol{\Cs}. It was first discovered in~\cite{pzqalt}. This connection is formulated as an algebraic characterization of \pol{\Cs} which holds when \Cs is an arbitrary \pvari of regular languages and whose statement is parameterized by \Cs-separation. Its main application is an effective reduction from \pol{\Cs}-membership to \Cs-separation. Thus, as soon as one designs a \Cs-separation algorithm, this yields ``for free'' a membership algorithm for the more complex class \pol{\Cs}.

Additionally, we present a second transfer theorem which applies to a smaller class than \pol{\Cs}: the intersection class \capol{\Cs}. This is the class containing all languages $L$ such that both $L$ and its complement belong to \pol{\Cs}. This second transfer theorem is a simple corollary of the first one ans was originally formulated in~\cite{AlmeidaBKK15}. However it is also stronger: it yields a reduction from $\capol{\Cs}$-membership to \Cs-{\bf membership}.

\section{Preliminary definitions}
\label{sec:prelims}
In this section, we fix the terminology and introduce several objects that we shall need to formulate and prove the results presented in the paper.

\subsection{Words and languages}

For the whole paper, we fix an arbitrary finite alphabet $A$. We denote by $A^*$ the set of all finite words over $A$, and by $\varepsilon \in A^*$ the empty word. Given two words $u,v \in A^*$, we write $u\cdot v$ (or simply $uv$) their concatenation. A \emph{language (over $A$)} is a subset of $A^*$. Abusing terminology, we denote by $u$ the singleton language $\{u\}$. It is standard to extend the concatenation operation to languages: given $K,L \subseteq A^*$, we write~$KL$ for the language $KL = \{uv \mid u \in K \text{ and } v \in L\}$. Moreover, we also consider marked concatenation, which is less standard. Given $K,L \subseteq A^*$, \emph{a marked concatenation} of $K$ with $L$ is a language of the form $KaL$ for some $a \in A$.

A class of languages \Cs is simply a set of languages. We say that \Cs is a \emph{lattice} when $\emptyset\in\Cs$, $A^* \in \Cs$ and \Cs is closed under union and intersection: for any $K,L \in \Cs$, we have $K \cup L \in \Cs$ and $K \cap L \in \Cs$. Moreover, a \emph{Boolean algebra} is a lattice \Cs which is additionally closed under complement: for any $L \in \Cs$, we have $A^* \setminus L \in \Cs$. Finally, a class \Cs is \emph{quotienting} if it is closed under quotients. That is, for any $L \in \Cs$ and any word $u \in A^*$, the following properties~hold:
\[
u^{-1}L \stackrel{\text{def}}{=}\{w\in A^*\mid uw\in L\} \text{\quad and\quad} Lu^{-1} \stackrel{\text{def}}{=}\{w\in A^*\mid wu\in L\}\text{\quad both belong to \Cs}.
\]
All classes that we consider are \varis of regular languages. These are the languages that can be equivalently defined by nondeterministic finite automata, finite monoids or monadic second-order logic. In the paper, we work with the definition by monoids, which we recall now.

\medskip
\noindent
{\bf Recognition by a monoid.} A \emph{monoid} is a set $M$ endowed with an associative multiplication $(s,t)\mapsto s\cdot t$ (we often write $st$ for $s\cdot t$) having a neutral element $1_M$, \emph{i.e.}, such that $1_M\cdot s=s\cdot 1_M=s$ for every $s \in M$. An \emph{idempotent} of a monoid $M$ is an element $e \in M$ such that $ee = e$. It is folklore that for any \emph{finite} monoid $M$, there exists a natural number $\omega(M)$ (denoted by $\omega$ when $M$ is understood) such that for any $s \in M$, the element $s^\omega$ is an idempotent.

We may now explain how to recognize languages with monoids. Observe that $A^{*}$ is a monoid whose multiplication is concatenation (the neutral element is $\varepsilon$). Thus, we may consider monoid morphisms $\alpha: A^* \to M$ where $M$ is an arbitrary monoid. Given such a morphism and some language $L \subseteq A^*$, we say that $L$ is \emph{recognized} by $\alpha$ when there exists a set $F \subseteq M$ such that $L = \alpha\inv(F)$. It is known that $L$ is regular if and only if it can be recognized by a morphism into a \textbf{finite} monoid.

Moreover, since we consider classes of languages that are not closed under complement (i.e. they are only lattices), we need to work with recognition by ordered monoids. An ordered monoid is a pair $(M,\leq)$ such that ``$\leq$'' is an order relation defined on $M$ which i compatible with its multiplication: given $s_1s_2,t_1,t_2 \in M$, if $s_1 \leq t_1$ and $s_2 \leq t_2$, then $s_1s_2 \leq t_1t_2$. Furthermore, we say that a subset $F \subseteq M$ is a \emph{upper set} for $\leq$ when given any $s \in F$ and any $t \in M$ such that $s \leq t$, we have $t \in F$ as well. Consider a morphism $\alpha: A^* \to M$ and ``$\leq$'' an order on $M$ such that $(M,\leq)$ is an ordered monoid. We say that some language $L \subseteq A^*$ is \emph{$\leq$-recognized} by $\alpha$ when there exists a \textbf{upper set} $F \subseteq M$ for $\leq$ such that $L = \alpha\inv(F)$.

\begin{remark}
	The key idea behind the definition is that the set of languages which are recognized by $\alpha: A^* \to M$ is necessarily closed under copmplement: if $L = \alpha\inv(F)$, then $A^* \setminus L = \alpha\inv(M \setminus F)$. However, this is not the case for the set of languages which are $\leq$-recognized by $\alpha$: while $F$ is an upper set, this need not be the case for $M \setminus F$.	 
\end{remark}

Finally, given any regular language $L$, one may define (and compute) a canonical morphism into a finite monoid which recognizes $L$: the syntactic morphism of $L$. Let us briefly recall its definition. One may associate to $L$ an equivalence $\equiv_L$ over $A^*$: the \emph{syntactic congruence of $L$}. Given $u,v \in A^*$, $u \equiv_L v$ if and only if $xuy \in L \Leftrightarrow xvy \in L$ for any $x,y \in A^*$. It is known and simple to verify that ``$\equiv_L$'' is a congruence on $A^*$. Thus, the set of equivalence classes $M_L = {A^*}/{\equiv_L}$ is a monoid and the map $\alpha_L: A^* \to M_L$ which maps any word to its equivalence class is a morphism. The monoid $M_L$ is called the syntactic monoid of $L$ and $\alpha_L$ its syntactic morphism. Finally, we may define a canonical order relation  ``$\leq_L$'' (called syntactic order) on the syntactic monoid $M_L$. Given $s,t \in M_L$, we write $s \leq_L t$ when for any $x,y \in M_L$, $xsy \in \alpha_L(L) \Rightarrow xty \in \alpha_L(L)$.  It is simple to verify that $(M_L,\leq_L)$ is an ordered monoid and that $L$ is $\leq_L$-recognized by $\alpha_L$.

It is known that $L$ is regular if and only if $M_L$ is finite (\emph{i.e.}, $\equiv_L$ has finite index): this is Myhill-Nerode theorem. In that case, one may compute the syntactic morphism $\alpha_L: A^* \to M_L$ (and the syntactic order on $M_L$) from any representation of $L$ (such as a finite automaton). 

\medskip
\noindent
{\bf Membership and separation.} In the paper, we are interested in two decision problems which we define now. Both are parameterized by some class of languages \Cs. Given a class of languages \Cs, the \Cs-membership problem is as follows:
\medskip

\begin{tabular}{rl}
	{\bf INPUT:}  &  A regular language $L$. \\
	{\bf OUTPUT:} &  Does $L$ belong to \Cs ? 
\end{tabular}
\medskip

Separation is slightly more involved. Given three languages $K,L_1,L_2$, we say that $K$ \emph{separates} $L_1$ from $L_2$ if $L_1 \subseteq L \text{ and } L_2 \cap K = \emptyset$. Given a class of languages \Cs, we say that $L_1$ is \emph{$\Cs$-separable} from $L_2$ if some language in \Cs separates $L_1$ from $L_2$. Observe that when \Cs is not closed under complement (which is the case for all classes investigated in the paper), the definition is not symmetrical: $L_1$ could be \Cs-separable from $L_2$ while $L_2$ is not \Cs-separable from $L_1$. The separation problem associated to a given class \Cs is as follows:

\medskip

\begin{tabular}{rl}
	{\bf INPUT:}  &  Two regular languages $L_1$ and $L_2$. \\
	{\bf OUTPUT:} &  Is $L_1$ $\Cs$-separable from $L_2$ ? 
\end{tabular}

\medskip

We use membership and separation as a mathematical tools for investigating classes of languages: given a fixed class \Cs, obtaining a \Cs-separation algorithm usually requires a solid understanding of~\Cs.

\subsection{Factorization forest theorem of Simon}

When proving our main theorem, we shall need the factorization forest theorem of Simon which is a combinatorial result about finite monoids. We briefly recall it here. We refer the reader to~\cite{kfacto,bfacto,cfacto} for more details and a proof.

Consider a finite monoid $M$ and a morphism $\alpha: A^* \rightarrow M$. An \emph{$\alpha$-factorization forest} is an ordered unranked tree whose nodes are labeled by words in $A^*$. For any inner node $x$ with label $w \in A^*$, if $w_1,\dots,w_n \in A^*$ are the labels of its children listed from left to right, then $w = w_1\cdots w_n$. Moreover, all nodes $x$ in the forest must be of the three following kinds:
\begin{itemize}
	\item \emph{Leaves} which are labeled by either a single letter or
	the empty word.
	\item \emph{Binary inner nodes} which have exactly two children.
	\item \emph{Idempotent inner nodes} which may have an arbitrary number of 	children. However, the labels $w_1,\dots,w_n$ of these children must satisfy $\alpha(w_1) = \cdots = \alpha(w_n) = e$ where $e$ is an idempotent element of $M$.
\end{itemize}
Note that an idempotent node with exactly two children is also a binary node. This is harmless.

Given a word $w \in A^*$, an \emph{$\alpha$-factorization forest for $w$} is an $\alpha$-factorization forest whose root is labeled by $w$. The \emph{height} of a factorization forest is the largest $h \in \nat$ such that it contains a branch with $h$ inner nodes (a single leaf has height $0$). We turn to the factorization forest theorem of Simon: there exists a bound depending only on $M$ such that any word admits an $\alpha$-factorization forest of height at most this bound.

\begin{theorem}[\cite{simonfacto,kfacto}] \label{thm:facto}
	Consider a morphism $\alpha: A^* \rightarrow M$. For all words $w \in A^*$, there exists an $\alpha$-factorization forest for $w$ of height at most $3|M|-1$.
\end{theorem}

\subsection{Finite lattices}

We finish the section with useful tools that we use to manipulate classes that are finite lattices (i.e. one that contains finitely many languages). Consider a finite lattice \Cs. One may associate a \emph{canonical preorder relation over $A^*$} to \Cs. The definition is as follows. Given $w,w' \in A^*$, we write $w \canoc w'$ if and only if the following holds:
\[
\text{For all $L \in \Cs$,} \quad w \in L \ \Rightarrow\ w' \in L.
\]
It is immediate from the definition that \canoc is transitive and reflexive, making it a preorder. The relation \canoc has many applications. We start with an important lemma, which relies on the fact that \Cs is finite. We say that a language $L \subseteq A^*$ is an \emph{upper set} (for \canoc) when for any two words $u,v \in A^*$, if $u \in L$ and $u \canoc v$, then $v \in L$.

\begin{lemma} \label{lem:canosatur}
	Let $\Cs$ be a finite lattice. Then, for any $L \subseteq A^*$, we have $L \in \Cs$ if and only if $L$ is an upper set for \canoc. In particular, \canoc has finitely many upper sets.
\end{lemma}

\begin{proof}
	Assume first that $L \in \Cs$. Then, for all $w \in L$ and all $w'$ such that $w \canoc w'$, we have $w' \in L$ by definition of \canoc. Hence, $L$ is an upper set.	Assume now that $L$ is an upper set. For any word $w$, we write $\uclos w$ for the  upper set $\uclos w = \{u \mid w \canoc u\}$. By definition of \canoc $\uclos w$ is the intersection of all $L \in \Cs$ such that $w \in L$. Therefore, $\uclos w \in \Cs$ since \Cs is a finite lattice (and is therefore closed under intersection). Finally, since $L$ is an upper set, we have,
	\[
	L = \bigcup_{w \in L} \uclos w.
	\]
	Hence, since \Cs is closed under union and is finite, $L$ belongs to \Cs.
\end{proof}

We complete this definition with another useful result. When \Cs is additionally closed under quotients, the canonical preorder \canoc is compatible with word concatenation. 

\begin{lemma} \label{lem:canoquo}
	Let \Cs be a \pvari. Then, the associated canonical preorder \canoc is compatible with word concatenation. That is, for any words $u,v,u',v'$,
	\[
	u \canoc u' \quad \text{and} \quad v \canoc v' \quad \Rightarrow \quad uv \canoc u'v'.
	\]		
\end{lemma}

\begin{proof}
	Let $u,u',v,v'$ be four words such that $u \canoc u'$ and $v \canoc v'$. We have to prove that $uv \canoc u'v'$. Let $L \in \Cs$ and assume that $uv \in L$. We use closure under left quotients to prove that $uv' \in L$ and then closure under right quotients to prove that $u'v' \in L$ which terminates the proof of this direction. Since $uv \in L$, we have $v \in u^{-1} \cdot L$. By closure under left quotients, we have $u^{-1} \cdot L \in \Cs$, hence, since $v \canoc v'$, we obtain that $v'\in u^{-1} \cdot L$ and therefore that $uv' \in L$. It now follows that $u \in L \cdot (v')^{-1}$. Using closure under right quotients, we obtain that $L \cdot (v')^{-1} \in \Cs$. Therefore, since $u \canoc u'$, we conclude that $u' \in L \cdot (v')^{-1}$ which means that $u'v' \in L$, as desired.
\end{proof}

\section{Polynomial closure}
\label{secpolc}
In this section, we define the polynomial closure operation defined on classes of languages. It is the main focus of the paper. We also prove a characteristic property of this operation that will be useful in proofs later.

\subsection{Definition}

Given an arbitrary class \Cs, the \emph{polynomial closure} of \Cs, denoted by \pol{\Cs}, is the smallest class containing \Cs and closed under marked concatenation and union: for any $H,L \in \pol{\Cs}$ and $a \in A$, we have $HaL \in \pol{\Cs}$ and $H \cup L \in \pol{\Cs}$.

It is not immediate that \pol{\Cs} has robust closure properties beyond those that are explicitly stated in the definitions. However, it turns out that when \Cs satisfies robust properties itself, this is the case for \pol{\Cs} as well. It was shown by Arfi~\cite{arfi91} that when \Cs is a \pvari of regular languages, then \pol{\Cs} is one as well. Note that this result is not immediate (the difficulty is to prove that \pol{\Cs} is closed under intersection).

\begin{theorem} \label{thm:polclos}
	Let \Cs be a \pvari of regular languages. Then, \pol{\Cs} is a \pvari of regular languages closed under concatenation and marked concatenation.
\end{theorem}

We shall obtain an alternate proof of Theorem~\ref{thm:polclos} as a corollary of our main result (i.e. our algebraic characterization of \pol{Cs}.

\medskip

Finally, we shall consider two additional operations which are defined by building on polynomial closure. Given a class \Cs, we denote by \copol{\Cs} the class containing all complements of languages in \pol{\Cs}: $L \in \copol{\Cs}$ when $A^* \setminus L \in \pol{\Cs}$. Finally, we also write \capol{\Cs} for the class of all languages that belong to both \pol{\Cs} and \copol{\Cs}. The following result is an immediate corollary of  Theorem~\ref{thm:polclos}.

\begin{corollary}
	Let \Cs be a \pvari of regular languages. Then, \copol{\Cs} is a \pvari of regular languages and \capol{\Cs} is a \vari of regular languages.
\end{corollary}

\begin{proof}
	By Theorem~\ref{thm:polclos}, \pol{\Cs} is a \pvari of regular languages. Since quotients commute with Boolean operations, it follows from De Morgan's laws that \copol{\Cs} is a \pvari of regular languages as well. Consequently, \capol{\Cs} is a \pvari of regular languages and since it must be closed under complement by definition, it is actually a \vari of regular languages. 
\end{proof}

\subsection{Characteristic property}

We complete the definitions with a property which applies to the polynomial closure of any \textbf{finite} \pvari \Cs. Recall that in this case, we associate a canonical preorder $\leq_{\Cs}$ over $A^*$ (two words are comparable when any language in \Cs containing the first word contains the second word as well). Since \Cs is closed under quotients, $\leq_\Cs$ must be compatible with word concatenation by Lemma~\ref{lem:canoquo}.

\begin{proposition} \label{prop:upol:mainprop}
	Let \Cs be a finite \pvari. Consider a language $L \subseteq A^*$ in \pol{\Cs}. Then, there exist natural numbers $h,p \geq 1$ such that for any $\ell \geq h$ and $u,v,x,y \in A^*$ satisfying $u \leq_\Cs v$, we have,
	\[
	xu^{p\ell+1} y \in L \quad \Rightarrow \quad xu^{p\ell} v u^{p\ell} y \in L
	\]
\end{proposition}

We now concentrate on proving Proposition~\ref{prop:upol:mainprop}. We fix the finite \pvari \Cs for the proof. Consider a language $L \subseteq A^*$ in \pol{\Cs}. We first need to choose the natural numbers $h,p \geq 1$ depending on $L$ and \Cs. We start by choosing $p$ with the following fact.

\begin{fact} \label{fct:period}
	There exists $p \geq 1$  such that for any $m,m' \geq 1$ and $w \in A^*$, $w^{pm} \leq_\Cs w^{pm'}$.
\end{fact}

\begin{proof}
	Let $\sim$ be the equivalence on $A^*$ generated by $\leq_{\Cs}$. Since $\leq_\Cs$ is a preorder with finitely many upper sets which is compatible with concatenation (see Lemma~\ref{lem:canosatur} and~\ref{lem:canoquo}), $\sim$ must be a congruence of finite index. Therefore, the set ${A^*}/{\sim}$ of $\sim$-classes if a finite monoid. It suffices to choose $p$ as the idempotent power of this finite monoid.
\end{proof}

It remains to choose $h$. Since $L$ belongs to $\pol{\Cs}$, it is built from languages in \Cs using only union and marked concatenations. It is simple to verify that these two operations commute. Hence, $L$ is a finite union of products having the form:
\[
L_0 a_1L_1 \cdots a_mL_m,
\]
where $a_1,\dots,a_m \in A$ and $L_0,\dots,L_m \in \Cs$. We define $n \in \nat$ as a natural number such that for any product $L_0 a_1L_1 \cdots a_mL_m$ in the union, we have $m \leq n$. Finally, we let,
\[
h = 2n+1
\]
It remains to show that $h$ and $p$ satisfy the desired property. Let $\ell \geq h$ and $u,v,x,y \in A^*$ satisfying $u \leq_\Cs v$. We have to show that,
\[
xu^{p\ell+1} y \in L \quad \Rightarrow \quad xu^{p\ell} v u^{p\ell} y \in L
\]
Consequently, we assume that $xu^{p\ell+1} y \in L$. By hypothesis, we know that there exists a product $L_0 a_1L_1 \cdots a_mL_m \subseteq L$ with $a_1,\dots,a_m \in A$, $L_0,\dots,L_m \in \Cs$ and $m \leq n$ such that $xu^{p\ell+1} y  \in L_0 a_1L_1 \cdots a_mL_m$. It follows that $xu^{p\ell+1} y$ admits a unique decomposition,
\[
xu^{p\ell+1} y  = w_0a_1w_1 \cdots a_m w_m
\]
such that $w_i \in L_i$ for all $i \leq m$. Recall that by definition $\ell \geq h =  2n+1 \geq 2m+1$. Therefore, it is immediate from a pigeon-hole principle argument that an infix $u^p$ of $xu^{p\ell+1} y$ must be contained within one of the infixes $w_i$. In other words, we have the following lemma.

\begin{lemma} \label{lem:upol:mainprop}
	There exist $i \leq m$, $j_1,j_2 < \ell$ such that $j_1+1+j_2 = \ell$ and $x_1,x_2 \in A^*$ satisfying,
	\begin{itemize}
		\item $w_i = x_1u^px_2$.
		\item $w_0a_1w_1 \cdots a_i x_1 = xu^{pj_1}$.
		\item $x_2 a_{i+1} \cdots a_m w_m = u^{pj_2+1} y$.				
	\end{itemize}
\end{lemma}

We may now finish the proof. By Fact~\ref{fct:period}, we have the following inequality,
\[
u^p \leq_\Cs u^{p(\ell+1)} = u^{p(j_1+1+j_2+1)} = u^{p(j_2+1)} u u^{p(j_1+1)-1}
\]
Moreover, since $u \leq_\Cs v$ and $\leq_\Cs$ is compatible with concatenation this yields that,
\[
u^p \leq_{\Cs} u^{p(j_2+1)} v u^{p(j_1+1)-1}
\]
Using again compatibility with concatenation we obtain,
\[
w_i = x_1u^px_2 \leq_{\Cs} x_1u^{p(j_2+1)} v u^{p(j_1+1)-1}x_2
\]
Therefore, since $w_i \in L_i$ which is a language of \Cs, it follows from the definition of $\leq_\Cs$ that $x_1u^{p(j_2+1)} v u^{p(j_1+1)-1}x_2 \in L_i$. Therefore, since $w_j \in L_j$ for all $j$,
\[
w_0a_1w_1 \cdots a_i x_1u^{p(j_2+1)} v u^{p(j_1+1)-1}x_2 a_{i+1} \cdots a_m w_m \in L_0 a_1L_1 \cdots a_mL_m
\]
By the last two items in Lemma~\ref{lem:upol:mainprop}, this exactly says that $xu^{p\ell} v u^{p\ell} y \in L_0 a_1L_1 \cdots a_mL_m$. Since we have $L_0 a_1L_1 \cdots a_mL_m \subseteq L$ by definition, this implies that $xu^{p\ell} v u^{p\ell} y \in L$, finishing the proof.

\section{Membership for \pol{\Cs}}
\label{sec:polc}
In this section, we prove the main theorem of the paper. Given an arbitrary \pvari of regular languages \Cs, \pol{\Cs}-membership reduces to \Cs-separation. We state this result in the following theorem.

\begin{theorem} \label{thm:trans:polreduc}
  Let \Cs be a \pvari of regular languages and assume that \Cs-separation is decidable. Then \pol{\Cs}-membership is decidable as well.
\end{theorem}

\begin{remark} \label{rem:trans:remtrans}
  Theorem~\ref{thm:trans:polreduc} is a generalization of a result from~\cite{pzqalt} which applies only to specific \pvaris \Cs belonging to a hierarchy of classes called the Straubing-Thérien hierarchy. However, let us point out that the main ideas behind the proof are all captured by the special case presented in~\cite{pzqalt}.
\end{remark}

This section is devoted to proving Theorem~\ref{thm:trans:polreduc}. It is based on an algebraic characterization of \pol{\Cs}. This characterization is formulated using equations on the syntactic ordered monoid of the language. These equations are parameterized by a relation on the syntactic monoid: the \emph{\Cs-pairs}. As we shall see, computing this relation requires an algorithm for \Cs-separation which explains the statement of Theorem~\ref{thm:trans:polreduc}.

We first present the definition of \Cs-pairs. We then use them to present the algebraic characterization of \pol\Cs and explain why Theorem~\ref{thm:trans:polreduc} is an immediate corollary. Finally, we then present a proof of this characterization. It relies on Simon's factorization forest theorem (Theorem~\ref{thm:facto}).

\subsection{\texorpdfstring{\Cs-pairs}{C-pairs}}

Consider a class of languages \Cs, an alphabet $A$, a finite monoid $M$ and a \emph{surjective} morphism $\alpha: A^* \to M$.  We define a relation on $M$: the \Cs-pairs (for $\alpha$). Consider a pair $(s,t) \in M \times M$. We say that,
\begin{equation} \label{def:trans:cpairs}
\text{$(s,t)$ is a \emph{\Cs-pair} (for $\alpha$) if and only if $\alpha\inv(s)$ is {\bf not} \Cs-separable from $\alpha\inv(t)$}
\end{equation}

\begin{remark}
	While we often make this implicit, being a \Cs-pair depends on the morphism $\alpha$.
\end{remark}

\begin{remark}
	While we restrict ourselves to \emph{surjective} morphisms, observe that the definition makes sense for arbitrary ones. We choose to make this restriction to ensure that we get a reflexive relation, which is not the case when $\alpha$ is not surjective  (if $s \in M$ has no antecedent $(s,s)$ is not a \Cs-pair). However this restriction is harmless: we use \Cs-pairs together with syntactic morphisms which are surjective.
\end{remark}

By definition, the set of \Cs-pairs for $\alpha$ is finite: it is a subset of $M \times M$. Moreover, having a \Cs-separation algorithm in hand is clearly enough to compute all \Cs-pairs for any input morphism $\alpha$.  While simple, this property is crucial, we state it in the following lemma.

\begin{lemma} \label{lem:trans:septopairs}
	Let \Cs be a class of languages and assume that \Cs-separation is decidable. Then, given an alphabet $A$, a finite monoid $M$ and a surjective morphism $\alpha: A^* \to M$ as input, one may compute all \Cs-pairs for $\alpha$.
\end{lemma}

We complete the definition with a few properties of \Cs-pairs. A simple and useful one is that the \Cs-pair relation is reflexive (it is not transitive in general).

\begin{lemma} \label{lem:trans:pairsreflex}
	Let \Cs be a class of languages, $A$ an alphabet, $M$ a finite monoid and $\alpha: A^* \to M$ a surjective morphism. Then, the \Cs-pair relation is reflexive: for any $s \in M$, $(s,s)$ is a \Cs-pair.
\end{lemma}

\begin{proof}
	Given $s \in M$, since $\alpha$ is surjective, we have $\alpha\inv(s) \neq \emptyset$. Therefore, $\alpha\inv(s) \cap \alpha\inv(s) \neq \emptyset$ and we obtain that $\alpha\inv(s)$ is not \Cs-separable from $\alpha\inv(s)$. This exactly says that $(s,s)$ is a \Cs-pair.
\end{proof}

Finally, we prove that when \Cs is a \pvari of regular languages (which is the only case that we shall consider), the \Cs-pair relation is compatible with multiplication.

\begin{lemma} \label{lem:trans:mult}
	Let \Cs be a \pvari of regular languages, $A$ an alphabet $M$ a finite monoid and $\alpha: A^* \to M$ a surjective morphism. For any two \Cs-pairs $(s_1,t_1),(s_2,t_2) \in M \times M$, $(s_1s_2,t_1t_2)$ is a \Cs-pair as well. 
\end{lemma}

\begin{proof}
	We prove the contrapositive. Assume that $(s_1s_2,t_1t_2)$ is not a \Cs-pair. We show that either $(s_1,t_1)$ is not a \Cs-pair or $(s_2,t_2)$ is not a \Cs-pair. By hypothesis, we have a separator $K \in \Cs$ such that $\alpha\inv(s_1s_2) \subseteq K$ and $K \cap \alpha\inv(t_1t_2) = \emptyset$. We define,
	\[
	H = \bigcap_{w \in \alpha\inv(s_2)} Kw\inv
	\]
	By definition, $H \in \Cs$ since \Cs is a \pvari and contains only regular languages (thus $K$ has finitely many right quotients by the Myhill-Nerode theorem)). Moreover, since $\alpha\inv(s_1s_2) \subseteq K$, one may verify from the definition that $\alpha\inv(s_1) \subseteq H$. There are now two cases. If $\alpha\inv(t_1) \cap H = \emptyset$ then $H \in \Cs$ separates $\alpha\inv(s_1)$ from $\alpha\inv(t_1)$ and we are finished: $(s_1,t_1)$ is not a \Cs-pair.  Otherwise, there exists a word $u \in \alpha\inv(t_1) \cap H \neq \emptyset$. Let $G = u\inv K \in \Cs$. We claim that $G$ separates $\alpha\inv(s_2)$ from $\alpha\inv(t_2)$ which concludes the proof: $(s_1,t_1)$ is not a \Cs-pair. Indeed, given $w \in \alpha\inv(s_2)$, we have $u \in H \subseteq Kw\inv$ which means that $uw \in K$ and therefore that $w \in G = u\inv K$. Moreover, assume by contradiction that there exists $v \in \alpha\inv(t_2) \cap G$. Since $G = u\inv K$, it follows that $uv \in K$. Finally, since $\alpha(u)=t_1$ and $\alpha(v)= t_2$, it follows that $uv \in \alpha\inv(t_1t_2)$. Thus, $uv \in K \cap \alpha\inv(t_1t_2)$ which is a contradiction since this language is empty by hypothesis.
\end{proof}

\subsection{Characterization theorem}

We now characterize of $\pol{\Cs}$ when \Cs is an arbitrary \pvari by a property of the syntactic morphism of the languages in \pol{\Cs}. As we announced, the characterization is parametrized by the \Cs-pair relation that we defined above.

\begin{theorem}\label{thm:trans:caracsig}
  Let $\Cs$ be a \pvari of regular languages and let $L$ be a regular language. Then, the three following properties are equivalent:
  \begin{enumerate}
  \item $L \in \pol{\Cs}$.
  \item The syntactic morphism $\alpha_L: A^* \to M_L$ of $L$ satisfies
    the following property:
    \begin{equation}\label{eq:trans:sig}
      s^{\omega+1} \leq_L s^{\omega}ts^{\omega} \quad \text{for all \Cs-pairs $(s,t) \in M_L^2$}.
    \end{equation}
  \item The syntactic morphism $\alpha_L: A^* \to M_L$ of $L$ satisfies the following property:
    \begin{equation}\label{eq:trans:sig2}
      e \leq_L ete \quad \text{for all \Cs-pairs $(e,t) \in M_L^2$ with $e$ idempotent}.
    \end{equation}
  \end{enumerate}
\end{theorem}

Theorem~\ref{thm:trans:caracsig} states a reduction from \pol{\Cs}-membership to \Cs-separation. Indeed, the syntactic morphism of a regular language can be computed and Equation~\eqref{eq:trans:sig} can be decided as soon as one is able to compute all $\Cs$-pairs (which is equivalent to deciding \Cs-separation by Lemma~\ref{lem:trans:septopairs}). Hence, we obtain Theorem~\ref{thm:trans:polreduc} as an immediate corollary. Moreover, Theorem~\ref{thm:polclos} is also a simple corollary of Theorem~\ref{thm:trans:caracsig}  (it is straightforward to verify that any class satisfying Item~\eqref{eq:trans:sig} in the theorem has to be a \pvari)

Moreover, observe that one may also use Theorem~\ref{thm:trans:caracsig} to obtain a symmetrical characterization for the class \copol{\Cs}. Recall that $\copol{\Cs}$ contains all languages whose complement is in \pol{\Cs}. It is straightforward to verify that a language and its complement have the same syntactic monoid but opposite syntactic orders. Therefore, we obtain the following corollary.

\begin{corollary} \label{cor:trans:caracpi}
  Let $\Cs$ be a \pvari of regular languages and let $L$ be a regular language. Then, the two following properties are equivalent:
  \begin{enumerate}
  \item $L \in \copol{\Cs}$.
  \item The syntactic morphism $\alpha_L: A^* \to M_L$ of $L$ satisfies the following property:
    \begin{equation}\label{eq:trans:pi}
      s^{\omega}ts^{\omega} \leq_L s^{\omega+1}  \quad \text{for all \Cs-pairs $(s,t) \in M_L^2$}.
    \end{equation}
  \item The syntactic morphism $\alpha_L: A^* \to M_L$ of $L$ satisfies the following property:
    \begin{equation}\label{eq:trans:pi2}
      ete \leq_L e \quad \text{for all \Cs-pairs $(e,t) \in M_L^2$ with $e$ idempotent}.
    \end{equation}
  \end{enumerate}
\end{corollary}

This terminates the presentation of the algebraic characterization of \pol{\Cs}. We now turn to its proof.

\subsection{Proof of Theorem~\ref{thm:trans:caracsig}}

We prove Theorem~\ref{thm:trans:caracsig}. Let \Cs be a \pvari of regular languages, and let us fix a regular language $L$. Let $\alpha_L: A^* \to M_L$ be its syntactic morphism. We prove that $1) \Rightarrow 2) \Rightarrow 3) \Rightarrow 1)$. We start with $1) \Rightarrow 2)$: when $L \in \pol{\Cs}$, $\alpha_L$ satisfies Equation~\eqref{eq:trans:sig}.

\subsubsection*{Direction $1) \Rightarrow 2)$}

Assume that $L \in \pol{\Cs}$. We have to show that $\alpha_L$ satisfies Equation~\eqref{eq:trans:sig}. Given a \Cs-pair $(s,t) \in M_L^2$, we have to show that $s^{\omega+1} \leq_L s^{\omega}ts^{\omega}$. We first prove the following simple fact.

\begin{fact} \label{fct:upol:upolstrat}
	There exists a finite \pvari $\Ds \subseteq \Cs$ such that $L \in \pol{\Ds}$.
\end{fact}

\begin{proof}
	Since $L \in \pol{\Cs}$, it is built from finitely many languages in \Cs using unions and marked concatenations. We let $\Fs \subseteq \Cs$ as the finite class containing all basic languages in \Cs used in the construction. Moreover, we let \Ds as the smallest \pvari containing \Fs. Clearly $\Ds \subseteq \Cs$ since \Cs is a \pvari itself. Moreover, $L \in \pol{\Ds}$ since \Ds contains all languages in \Cs required to build $L$ by definition. It remains to show that \Ds remains finite. By definition, the languages in \Ds are built from those in \Fs by applying unions and intersections. Therefore, since quotients commute with Boolean operations, any language in \Ds is built by applying intersections and unions to languages in \Fs. Finally, any regular language has finitely many quotients by Myhill-Nerode theorem. Thus, since \Fs was finite, this is the case for \Ds as well.
\end{proof}

We work with the canonical preorder $\leq_\Ds$ over $A^*$ associated to the finite \pvari \Ds. Since $(s,t)$ is a \Cs-pair, we know that $\alpha\inv(s)$ is not \Cs-separable from $\alpha\inv(t)$. Therefore, since $\Ds \subseteq \Cs$, it follows that $\alpha\inv(s)$ is not \Ds-separable from $\alpha\inv(t)$. Consider the language,
\[
H = \{v \in A^* \mid u \leq_\Ds v \text{ for some $u \in \alpha\inv(s)$}\}
\]
By definition, $H$ is an upper set for $\leq_\Ds$ and therefore belongs to \Ds by  Lemma~\ref{lem:canosatur}. Moreover, $H$ includes $\alpha\inv(s)$ by definition. Consequently, since $\alpha\inv(s)$ is not \Ds-separable from $\alpha\inv(t)$, we know that $H$ intersects $\alpha\inv(t)$. This yields $u \in \alpha\inv(s)$ and $v \in \alpha\inv(t)$ such that $u \leq_\Ds v$. Hence, we may apply Proposition~\ref{prop:upol:mainprop} which yields natural numbers $h,p \geq 1$ such that for any $x,y \in A^*$,
\[
xu^{ph\omega+1} y \in L \quad \Rightarrow \quad xu^{ph\omega}v u^{ph\omega} y \in L
\]
By definition of the syntactic order on $M_L$, it then follows that,
\[
s^{\omega+1} = \alpha(u^{ph\omega+1}) \leq_L \alpha(u^{ph\omega}v u^{ph\omega}) = s^\omega t s^\omega
\]
This concludes the proof for this direction.
                                
\subsubsection*{Direction $2) \Rightarrow 3)$}

Let us assume that the syntactic morphism $\alpha_L: A^* \to M_L$ of $L$ satisfies~\eqref{eq:trans:sig}. We need to prove that it satisfies~\eqref{eq:trans:sig2} as well. Let $(e,t) \in M_L^2$ be a \Cs-pair with $e$ idempotent. We have to show that $e \leq_L ete$. Since~\eqref{eq:trans:sig} holds, we know that $e^{\omega+1} \leq_L e^{\omega}te^{\omega}$. Moreover, since $e$ is idempotent, we have $e = e^{\omega+1} = e^{\omega}$. Thus, we get $e \leq_L ete$ as desired.                                            

\subsubsection*{Direction $3) \Rightarrow 1)$}

It now remains to prove the harder ``$3) \Rightarrow 1)$'' direction of Theorem~\ref{thm:trans:caracsig}. We use induction to prove that for any finite ordered monoid $(M,\leq)$ and any surjective morphism $\alpha: A^* \to M$ satisfying~\eqref{eq:trans:sig2}, any language $\leq$-recognized by $\alpha$ may be constructed from languages of \Cs using unions and (marked) concatenations (thus showing that it belongs to \pol{\Cs}). Since $L$ is $\leq_L$-recognized by its syntactic morphism, this ends the proof.

We fix a surjective morphism $\alpha: A^* \to M$ satisfying~\eqref{eq:trans:sig2}: for any \Cs-pair $(e,t) \in M^2$ with $e$ idempotent, we have $e \leq ete$. The proof is based on Simon's factorization forest theorem (see Section~\ref{sec:prelims}). We state it in the following proposition.

\begin{proposition} \label{prop:trans:signec2}
  For all $h \in \nat$ and all $s \in M$, there exists $H_{s,h} \in \pol{\Cs}$ such that for all $w \in A^*$:
  \begin{itemize}
  \item If $w \in H_{s,h}$ then $s \leq \alpha(w)$.
  \item If $\alpha(w) = s$ and $w$ admits an $\alpha$-factorization forest of height at most $h$ then $w \in H_{s,h}$.
  \end{itemize}
\end{proposition}

Assume for now that Proposition~\ref{prop:trans:signec2} holds. Given $h = 3|M|-1$, for all $s \in M$, consider the language $H_{s,h} \in \pol{\Cs}$ associated to $s$ and $h$ by Proposition~\ref{prop:trans:signec2}. We know from Simon's Factorization Forest theorem (Theorem~\ref{thm:facto}) that all words in $A^*$ admit an $\alpha$-factorization forest of height at most $3|M| - 1$. Therefore, for all $w \in A^*$ we have,
\begin{enumerate}
\item\label{item:trans:4} If $w \in H_{s,h}$ then $s \leq \alpha(w)$.
\item\label{item:trans:5} If $\alpha(w) = s$ then $w \in H_{s,h}$.
\end{enumerate}
Let $L$ be some language $\leq$-recognized by $\alpha$ and let $F$ be its accepting set. Observe that $L = \bigcup_{s \in F} H_{s,h}$. Indeed, by Item~\ref{item:trans:5} above, we have $L \subseteq \bigcup_{s \in F} H_{s,h}$. Moreover, by definition of $\leq$-recognizability, $F$ has to be an upper set, that is, if $s \in F$ and $s \leq t$ then $t\in F$. Hence, Item~\ref{item:trans:4} above implies that $\cup_{s \in F} H_{s,h} \subseteq L$. We conclude that $L \in \pol{\Cs}$ since it is a union of languages $H_{s,h} \in \pol{\Cs}$. This finishes the proof of Theorem~\ref{thm:trans:caracsig}. It now remains to prove Proposition~\ref{prop:trans:signec2}.

We begin with a lemma which defines the basic languages in \Cs that we will use in the construction of our languages in \pol{\Cs}. Note that this is also where we use the fact that~\eqref{eq:trans:sig2} holds.

\begin{lemma} \label{lem:trans:kisright}
  For any idempotent $e \in M$, there exists a language $K_e$ belonging to \Cs (and therefore to \pol{\Cs}) which satisfies the two following properties,
  \begin{enumerate}
  \item For all $u \in K_e$, we have $e \leq e\alpha(u)e$.
  \item $\alpha\inv(e) \subseteq K_e$.
  \end{enumerate}
\end{lemma}

\begin{proof}
  Let $T \subseteq M$ be the set of all elements $t \in M$ such that $(e,t)$ is {\bf not} a \Cs-pair (\emph{i.e.}, $\alpha\inv(e)$ is \Cs-separable from $\alpha\inv(t)$). By definition, for all $t \in T$, there exists a language $G_{t} \in \Cs$ which separates $\alpha\inv(e)$ from $\alpha\inv(t)$. We let $K_e = \bigcap_{t \in T} G_{t}$. Clearly, $K_e \in \Cs$ since \Cs is a \pvari, and is therefore closed under intersection. Moreover, $\alpha\inv(e) \subseteq K_e$ since the inclusion holds for all languages $G_{t}$. Finally, given $u \in K_e$, it is immediate from the definition that $\alpha(u)$ does not belong to $T$ which means that $(e,\alpha(u))$ is a \Cs-pair. The first item is now immediate from~\eqref{eq:trans:sig2} since $e$ is idempotent.
\end{proof}

We may now start the proof of Proposition~\ref{prop:trans:signec2}. Let $h \geq 1$ and $s \in M$. We construct $H_{s,h} \in \pol{\Cs}$ by induction on $h$. Assume first that $h = 0$. Note that the nonempty words having an $\alpha$-factorization forest of height at most $0$ are all single letters. We let $B = \{b \in A \mid \alpha(b) = s\}$. Moreover, we use the language $K_{1_{M}}$ as defined in Lemma~\ref{lem:trans:kisright} for the neutral element $1_{M}$ (which is an idempotent). There are two cases depending on whether $s = 1_{M}$ or not. If $s \neq 1_{M}$, we let,
\[
  H_{s,0} = \bigcup_{b \in B} K_{1_{M}}bK_{1_{M}}.
\]
Otherwise, when $s = 1_{M}$, we let,
\[
  H_{s,0} = K_{1_{M}} \cup \bigcup_{b \in B} K_{1_{M}}bK_{1_{M}}.
\]
Note that $H_{s,0} \in \pol{\Cs}$ since we only used marked concatenation and unions and $K_{1_{M}}  \in \Cs \subseteq \pol{\Cs}$ by definition in Lemma~\ref{lem:trans:kisright}. We now prove that this definition satisfies the two conditions in Proposition~\ref{prop:trans:signec2}. We do the proof for the case when $s \neq 1_{M}$ (the other case is similar).

Assume first that $w \in H_{s,0}$,  we have to prove that $s \leq \alpha(w)$. By definition $w = ubu'$ with $u,u' \in K_{1_{M}}$ and $b \in B$. Hence, $\alpha(w) = \alpha(u)s\alpha(u')$. Since $u,u' \in K_{1_{M}}$, we obtain from the second item in Lemma~\ref{lem:trans:kisright} that $1_{M} \leq \alpha(u)$ and $1_{M} \leq \alpha(u')$. It follows that $s \leq \alpha(u)s\alpha(u') = \alpha(w)$.

We turn to the second item. Let $w \in A^*$ such that $\alpha(w) = s$ and $w$ admits an $\alpha$-factorization forest of height at most $0$. Since we assumed that $s \neq 1_M$, $w$ cannot be empty. We have to prove that $w \in H_{s,0}$. By hypothesis, $w$ is a one letter word $b \in B$. % since it admits an $\alpha$-factorization forest of height at most $1$.
Hence, $w \in K_{1_{M}}bK_{1_{M}}$ since $\varepsilon \in K_{1_{M}}$ by the first item in Lemma~\ref{lem:trans:kisright}.

\medskip

Assume now that $h > 0$. There are two cases depending on whether $s$ is idempotent or not. We treat the idempotent case (the other case is essentially a simpler version of the same proof). Hence, we assume that $s$ is an idempotent, that we denote by $e$. We begin by constructing $H_{e,h}$ and then prove that it satisfies the conditions in the proposition. For all $t \in M$, one can use induction to construct $H_{t,h-1} \in \pol{\Cs}$ such that for all $w \in A^*$:
\begin{itemize}
\item If $w \in H_{t,h-1}$ then $t \leq \alpha(w)$.
\item If $\alpha(w) = t$ and $w$ is empty or admits an $\alpha$-factorization
  forest of height at most $h-1$, then $w \in H_{t,h-1}$.
\end{itemize}

We now define $H_{e,h}$ as the union of three languages. Intuitively, the first one contains the words which are either empty or have an $\alpha$-factorization forest of height at most $h-1$, the second one, words having an $\alpha$-factorization forest of height $h$ and whose root is a binary node, and the third one, words with an $\alpha$-factorization forest of height $h$ and whose root is an idempotent node.
\[
  H_{e,h} = H_{e,h-1} \ \cup\ \bigcup_{t_1t_2=e} (H_{t_1,h-1}H_{t_2,h-1})\ \cup\ H_{e,h-1}K_eH_{e,h-1} \quad \text{with $K_e$ as defined in Lemma~\ref{lem:trans:kisright}}
\]
Note that by definition, $H_{e,h}$ is a union of concatenations of languages in \pol{\Cs} and therefore belongs to \pol{\Cs} itself. We need to prove that it satisfies the conditions of the proposition. Choose some $w \in A^*$ and assume first that $w \in H_{e,h}$. We need to prove that $e \leq \alpha(w)$.
\begin{itemize}
\item If $w \in H_{e,h-1}$, then this is by definition of $H_{e,h-1}$.
\item If $w \in H_{t_1,h-1}H_{t_2,h-1}$ for $t_1,t_2 \in M$ such that $t_1t_2 = e$, then by definition, $w = w_1w_2$ with $t_1 \leq \alpha(w_1)$ and ${t_2} \leq \alpha(w_2)$. It follows that $e = t_1t_2 \leq \alpha(w_1w_2) = \alpha(w)$.
\item Finally, if $w \in H_{e,h-1}K_eH_{e,h-1}$, we obtain that $w = w_1uw_2$ with $e \leq \alpha(w_1)$, $u \in K_e$ and $e \leq \alpha(w_2)$. In particular, by the second item in Lemma~\ref{lem:trans:kisright}, $e \leq e\alpha(u)e$. Hence, since $e\alpha(u)e \leq \alpha(w_1)\alpha(u)\alpha(w_2) = \alpha(w)$, we conclude that $e \leq \alpha(w)$.
\end{itemize}

Conversely, assume that $\alpha(w) = e$ and that $w$ admits an $\alpha$-factorization forest of height at most $h$. We have to prove that $w \in H_{e,h}$. There are again three cases.
\begin{itemize}
\item First, if $w$ is empty or admits an $\alpha$-factorization forest of height at most
  $h-1$, then $w \in H_{e,h-1}$ by definition.
\item Second, if $w$ admits an $\alpha$-factorization forest of height $h$ whose root is a binary node, then $w = w_1w_2$ with $w_1,w_2$ admitting forests of height at most $h-1$. Let $t_1= \alpha(w_1)$ and ${t_2} = \alpha(w_2)$. Observe that $t_1t_2 = \alpha(w) = e$. By the definition, we have $w_1 \in H_{t_1,h-1}$ and $w_2 \in H_{t_2,h-1}$. Hence, $w \in H_{t_1,h-1}H_{t_2,h-1} \subseteq H_{e,h}$ and we are finished.
\item Finally, if $w$ admits an $\alpha$-factorization forest of height $h$ whose root is an idempotent node, then $w = w_1uw_2$ with $\alpha(w_1) = \alpha(u) = \alpha(w_2) = e$ and $w_1,w_2$ admitting forests of height at most $h-1$. It follows that $w_1,w_2 \in H_{e,h-1}$ and since $\alpha(u) = e$, it is immediate that $u \in K_e$ by first item in Lemma~\ref{lem:trans:kisright}. We conclude that $w \in H_{e,h-1}K_eH_{e,h-1} \subseteq H_{e,h}$.
\end{itemize}

This concludes the proof of Proposition~\ref{prop:trans:signec2}.

%%% Local Variables:
%%% mode: latex
%%% TeX-master: "../main"
%%% End:

\section{\texorpdfstring{Membership for \capol{\Cs}}{Membership for the intersection between Pol(C) and co-Pol(C)}}
\label{sec:capolc}
In this last section, we present a second transfer theorem which applies to the intersection class \capol{\Cs}. Recall that this denotes the class made of all languages which belong to both \pol{\Cs} and \copol{\Cs}. 

The membership problem is simpler to handle for \capol{\Cs} than it is for \pol{\Cs}. Recall that using the generic characterization of \pol{\Cs} (i.e. Theorem~\ref{thm:trans:caracsig}) to decide \pol{\Cs}-membership requires an algorithm for \Cs-separation. In other words, we reduced \pol{\Cs}-membership to a stronger problem for \Cs: separation. It turns out that deciding membership for \capol{\Cs} only requires an algorithm for \Cs-{\bf membership}: the same problem is used on both ends of the reduction. Intuitively, this second transfer result is much stronger than the previous one. However, it turns out that the former is a simple corollary of the latter: it is obtained via a few algebraic manipulations on the generic characterization of \pol{\Cs} (i.e. Theorem~\ref{thm:trans:caracsig}). This was first observed by Almeida, Bartonov{\'{a}}, Kl{\'{\i}}ma and Kunc~\cite{AlmeidaBKK15}.

\begin{theorem}\label{thm:trans:polcopolreduc}
	Let \Cs be a \pvari of regular languages and assume that \Cs-membership is decidable. Then $(\capol{\Cs})$-membership is decidable as well.
\end{theorem}

This section is devoted to proving Theorem~\ref{thm:trans:polcopolreduc}. Similarly to Theorem~\ref{thm:trans:polreduc},  the argument is based on an algebraic characterization of \capol{\Cs} parametrized by a relation depending on \Cs. However, unlike the \Cs-pairs that we used  in the \pol{\Cs}-characterization (i.e. Theorem~\ref{thm:trans:caracsig}), this new relation can be computed as soon as \Cs-membership is decidable. We speak of \emph{saturated \Cs-pairs}. We first define this new object and then use it to present the characterization of \capol{\Cs}.

\subsection{\texorpdfstring{Saturated \Cs-pairs}{Saturated C-pairs}}

Consider a class of languages \Cs, an alphabet $A$, a finite monoid $M$ and a \emph{surjective} morphism $\alpha: A^*\to M$. We define a new relation on $M$: the \emph{saturated \Cs-pairs} (for $\alpha$). Consider a pair $(s,t) \in M \times M$. We say that,
\begin{equation} \label{def:trans:scpairs}\begin{array}{c}
\text{$(s,t)$ is a saturated \emph{\Cs-pair} (for $\alpha$)} \\
 \text{if and only if} \\
\text{{\bf no} language $K \in \Cs$ {\bf recognized by $\alpha$} separates $\alpha\inv(s)$ from $\alpha\inv(t)$}
\end{array}
\end{equation}

Clearly, this new notion is closely related to the \Cs-pairs that we defined in Section~\ref{sec:polc}. When $(s,t)$ is a \Cs-pair, $\alpha\inv(s)$ is not \Cs-separable from $\alpha\inv(t)$. This means that no language $K \in \Cs$ (including those recognized by $\alpha$) separates $\alpha\inv(s)$ from $\alpha\inv(t)$. Thus, $(s,t)$ is also a saturated \emph{\Cs-pair}.

\begin{fact} \label{fct:trans:satunsat}
	Consider a class \Cs, an alphabet $A$, a finite monoid $M$ and a surjective morphism $\alpha: A^* \to M$. Then, any \Cs-pair $(s,t) \in M \times M$ is also a saturated \Cs-pair.	 
\end{fact}

\begin{remark}
	The converse of Fact~\ref{fct:trans:satunsat} is false in general: an arbitrary saturated \Cs-pair need not be a \Cs-pair. Indeed, we shall later prove that the saturated \Cs-pair relation is transitive and we already stated that the \Cs-pair relation is not. In fact, we prove below that the saturated \Cs-pairs are exactly the transitive closure of the original \Cs-pairs.
\end{remark}

While very similar to \Cs-pairs, saturated \Cs-pairs are also simpler to handle. In particular, having an algorithm for \Cs-membership suffices to compute all saturated \Cs-pairs. Indeed, with such a procedure in hand, it is possible to compute all subsets $F \subseteq M$ such that $\alpha\inv(F) \in \Cs$. One may then decide whether $(s,t) \in M \times M$ is a saturated \Cs-pair by checking whether one of these subsets $F$ satisfies $s \in F$ and $t \not\in F$. We state this in the following lemma.

\begin{lemma} \label{lem:trans:membtopairs}
	Let \Cs be a class of languages and assume that \Cs-membership is decidable. Then, given an alphabet $A$, a finite monoid $M$ and a surjective morphism $\alpha: A^* \to M$ as input, one may compute all saturated \Cs-pairs for $\alpha$.
\end{lemma}

Furthermore, saturated \Cs-pairs satisfy stronger properties than the original \Cs-pairs: they correspond to a \emph{transitive relation}. Altogether, this means that the saturated \Cs-pair relation is a preorder for an arbitrary class \Cs.

\begin{lemma} \label{lem:trans:satquot}
	Let \Cs be a class of languages, $A$ an alphabet, $M$ a finite monoid and $\alpha: A^* \to M$ a surjective morphism. Then, the three following properties hold:
	\begin{itemize}
		\item The saturated \Cs-pair relation is reflexive: for any $s \in M$, $(s,s)$ is a saturated \Cs-pair.
		\item The saturated \Cs-pair relation is transitive: for any $r,s,t \in M$ such that $(r,s)$ and $(s,t)$ are saturated \Cs-pairs, $(r,t)$ is a saturated \Cs-pair as well.
	\end{itemize} 
\end{lemma}

\begin{proof}
	For the first item, we know from Lemma~\ref{lem:trans:pairsreflex} that for any $s \in M$, $(s,s)$ is a \Cs-pair. Therefore, it is also a saturated \Cs-pair by Fact~\ref{fct:trans:satunsat}.
	
	We turn to the second item. Consider $r,s,t \in M$ such that $(r,s)$ and $(s,t)$ are saturated \Cs-pairs. We show that $(r,t)$ is a saturated \Cs-pair as well. That is, we must show that no language of \Cs recognized by $\alpha$ separates $\alpha\inv(r)$ from $\alpha\inv(t)$. Thus, consider $L \in \Cs$ recognized by $\alpha$ such that $\alpha\inv(r) \subseteq L$. We have to show that $\alpha\inv(t) \cap L = \emptyset$. Since $(r,s)$ is a saturated \Cs-pair, $L$ cannot separate $\alpha\inv(r)$ from $\alpha\inv(s)$. Thus, $\alpha\inv(s) \cap L = \emptyset$. Moreover, since $L$ is recognized by $\alpha$, this implies that $\alpha\inv(s) \subseteq L$. Finally, since $(s,t)$ is a saturated \Cs-pair, $L$ cannot separate $\alpha\inv(s)$ from $\alpha\inv(t)$. Thus, $\alpha\inv(t) \cap L = \emptyset$ and we are finished.
\end{proof}

Another useful property is that the saturated \Cs-pairs characterize exactly the languages in \Cs which are also recognized by the morphism $\alpha$ (provided that \Cs is a lattice).

\begin{lemma} \label{lem:trans:satpairdef}
	Let \Cs be a lattice, $A$ an alphabet, $M$ a finite monoid and $\alpha: A^* \to M$ a surjective morphism. Then, for any $F \subseteq M$, the two following properties are equivalent:
	\begin{enumerate}
		\item $\alpha\inv(F) \in \Cs$.
		\item $F$ is a upper set for the saturated \Cs-pair relation: for any $s \in F$ and any $t \in M$ such that $(s,t)$ is a saturated \Cs-pair, we have $t \in F$.
	\end{enumerate}	 
\end{lemma}

\begin{proof}
	We start with the direction $(1) \Rightarrow (2)$. Assume that $\alpha\inv(F) \in \Cs$. Consider $s \in F$ and $t \in M$ such that $(s,t)$ is a saturated \Cs-pair, we show that $t \in F$. We proceed by contradiction, assume that $t \not\in F$. In that case it is immediate that $\alpha\inv(F)$ separates $\alpha\inv(s)$ from $\alpha\inv(t)$. Since we have $\alpha\inv(F) \in \Cs$, this contradicts the hypothesis that $(s,t)$ is a saturated \Cs-pair and we are finished.
	
	We turn to the direction $(2) \Rightarrow (1)$. Assume that for any $s \in F$ and any $t \in M$ such that $(s,t)$ is a saturated \Cs-pair, we have $t \in F$. We show that $\alpha\inv(F) \in \Cs$. Consider $s \in F$ and $r \not\in F$. By hypothesis, we know that $(s,r)$ is {\bf not} a saturated \Cs-pair. Thus, we have $G_{s,r} \subseteq M$ such that $\alpha\inv(G_{s,r}) $ belongs to \Cs and separates $\alpha\inv(s)$ from $\alpha\inv(r)$. One may then verify that,
	\[
	\alpha\inv(F) = \bigcup_{s \in F} \bigcap_{r \not\in F} \alpha\inv(G_{s,r}) 
	\]
	Since $\Cs$ is a lattice, follows that $\alpha\inv(F) \in \Cs$. This concludes the proof.
\end{proof}

We may now further connect the saturated \Cs-pair relation with original \Cs-pair relation. We show that the former is the transitive closure of the latter.

\begin{lemma} \label{lem:trans:transclos}
	Consider a lattice \Cs, an alphabet $A$, a finite monoid $M$ and a surjective morphism $\alpha: A^* \to M$. Then, for any $(s,t) \in M \times M$, the following properties are equivalent, 
	\begin{enumerate}
		\item $(s,t)$ is a saturated \Cs-pair.
		\item There exist $n \in \nat$ and $r_0,\dots,r_{n+1} \in M$ such that $r_0 = s$, $r_{n+1} = t$ and $(r_{i},r_{i+1})$ is a \Cs-pair for all $i \leq n$.
	\end{enumerate}
\end{lemma}

\begin{proof}
	We already proved the direction $(2) \Rightarrow (1)$. Indeed, we know from Fact~\ref{fct:trans:satunsat} that any \Cs-pair is also a saturated \Cs-pair. Moreover, we showed in Lemma~\ref{lem:trans:transclos} that the saturated \Cs-pair relation is transitive. Therefore, we concentrate on the direction $(1) \Rightarrow (2)$. Let $(s,t)$ be a saturated \Cs-pair. Let $F \subseteq M$ as the smallest subset of $M$ satisfying the two following properties:
	\begin{enumerate}
		\item $s \in F$.
		\item For any \Cs-pair $(u,v) \in M \times M$, if $u \in F$, then $v \in F$ as well. 
	\end{enumerate}
	We have $s \in F$ by definition. We show that $\alpha\inv(F) \in \Cs$. By Lemma~\ref{lem:trans:satpairdef}, this will imply that $t \in F$ as well since $(s,t)$ is a saturated \Cs-pair. Thus, $(2)$ holds.
	
	Observe that for any $u \in F$, we may build a language $H_u \in \Cs$ such that $\alpha\inv(u) \subseteq H_u \subseteq \alpha\inv(F)$. Indeed, for any $v \not\in F$, we know that $(u,v)$ is not a \Cs-pair by definition of $F$. Thus, we have $H_{u,v} \in \Cs$ which separates $\alpha\inv(u)$ from $\alpha\inv(v)$. We may now define,
	\[
	H_u = \bigcap_{v \not \in F} H_{u,v}
	\]
	Clearly $H_u \in \Cs$ since \Cs is a lattice. It now suffices to observe that,
	\[
	\alpha\inv(F) = \bigcup_{u \in F} \alpha\inv(u) \subseteq \bigcup_{u \in F} H_u \subseteq \alpha\inv(F)
	\]
	Thus, $\alpha\inv(F) =\bigcup_{u \in F} H_u$ belong to \Cs since \Cs is lattice.
\end{proof}

Finally, we prove that when \Cs is a \pvari the saturated \Cs-pair relation is compatible with multiplication.

\begin{lemma} \label{lem:trans:smult}
	Let \Cs be a \pvari of regular languages, $A$ an alphabet $M$ a finite monoid and $\alpha: A^* \to M$ a surjective morphism. For any two saturated \Cs-pairs $(s_1,t_1),(s_2,t_2) \in M \times M$, $(s_1s_2,t_1t_2)$ is a saturated \Cs-pair as well. 
\end{lemma}

\begin{proof}
	Immediate from Lemma~\ref{lem:trans:transclos} since we already know that the \Cs-pair relation is compatible with multiplication by Lemma~\ref{lem:trans:mult}.
\end{proof}

\subsection{Characterization theorem}

We may now present the announced algebraic characterization of \capol{\Cs} and use it to prove Theorem~\ref{thm:trans:polcopolreduc}.

\begin{theorem} \label{thm:trans:caracint}
	Let \Cs be a \pvari of regular languages and $L$ a regular language. Then, the three following properties are equivalent:
	\begin{enumerate}
		\item $L \in \capol{\Cs}$.
		\item The syntactic morphism $\alpha_L: A^* \to M_L$ of $L$ satisfies
		the following property:
		\begin{equation}
		s^{\omega+1} = s^{\omega}ts^{\omega} \quad \text{for all \Cs-pairs $(s,t) \in M_L^2$} \label{eq:trans:caraint}
		\end{equation}
		\item The syntactic morphism $\alpha_L: A^* \to M_L$ of $L$ satisfies
		the following property:
		\begin{equation}
		s^{\omega+1} = s^{\omega}ts^{\omega} \quad \text{for all saturated \Cs-pairs $(s,t) \in M_L^2$} \label{eq:trans:caraint2}
		\end{equation}
	\end{enumerate}	
\end{theorem}

As announced, Theorem~\ref{thm:trans:caracint} states a reduction from $(\capol{\Cs})$-membership to \Cs-membership. Indeed, the syntactic morphism of a regular language can be computed and Equation~\eqref{eq:trans:caraint2} can be decided as soon as one is able to compute all saturated $\Cs$-pairs (as we explained, this amounts to deciding \Cs-membership). Hence, we obtain Theorem~\ref{thm:trans:polcopolreduc} as an immediate corollary. We turn to the proof of Theorem~\ref{thm:trans:caracint}.

\begin{proof}[Proof of Theorem~\ref{thm:trans:caracint}]
The equivalence $(1) \Leftrightarrow (2)$ follows from Theorem~\ref{thm:trans:caracsig} and Corollary~\ref{cor:trans:caracpi}. Indeed, by definition $L \in \capol{\Cs}$ if and only if $L \in \pol{\Cs}$ and $L \in \copol{\Cs}$. By Theorem~\ref{thm:trans:caracsig} and Corollary~\ref{cor:trans:caracpi} respectively, this is equivalent to $\alpha_L$ satisfying the two following properties:
\[
\begin{array}{lll}
s^{\omega+1} & \leq_L & s^{\omega}ts^{\omega} \quad \text{for all \Cs-pairs $(s,t) \in M^2$} \\
s^{\omega+1} & \geq_L & s^{\omega}ts^{\omega} \quad \text{for all \Cs-pairs $(s,t) \in M^2$}
\end{array}
\]
Clearly, when put together, these two equations are equivalent to~\eqref{eq:trans:caraint}. This concludes the proof of $(1) \Leftrightarrow (2)$.

\medskip

We now show that $(2) \Leftrightarrow (3)$. The direction $(3) \Rightarrow (2)$ is immediate from Fact~\ref{fct:trans:satunsat}. Indeed, since any \Cs-pair is also a saturated \Cs-pair, it is immediate that when~\eqref{eq:trans:caraint2} holds, then~\eqref{eq:trans:caraint} holds as well. Therefore, we concentrate on the direction $(2) \Rightarrow (3)$. We assume that~\eqref{eq:trans:caraint} holds and prove that this is the case for~\eqref{eq:trans:caraint2} as well. Consider a saturated \Cs-pair $(s,t) \in M_L^2$. We have to show that $s^{\omega+1} = s^{\omega}ts^{\omega}$.

By Lemma~\ref{lem:trans:transclos}, we know that there exist $n \in \nat$ and $r_0,\dots,r_{n+1} \in M$ such that $r_0 = s$, $r_{n+1} = t$ and $(r_{i},r_{i+1})$ is a \Cs-pair for all $i \leq n$. We prove by induction that for all $1 \leq k \leq n+1$, we have,
\[
s^{\omega+1} = s^{\omega}r_ks^{\omega}
\]
The case $k = n+1$ yields the desired result since $r_{n+1} = t$. When $k = 1$, it is immediate from~\eqref{eq:trans:caraint} that $s^{\omega+1} = s^{\omega}r_1s^{\omega}$ since $(s,r_1)$ is a \Cs-pair. We now assume that $k > 1$. Using induction, we get that,
\[
s^{\omega+1} = s^{\omega}r_{k-1}s^{\omega}
\]
Therefore, we obtain,
\[
s^{\omega} = (s^{\omega+1})^\omega 
= (s^{\omega}r_{k-1}s^{\omega})^\omega
\]
Since $(r_{k-1},r_k)$ is a \Cs-pair, It is immediate from Lemma~\ref{lem:trans:satquot} that, $(s^{\omega}r_{k-1}s^{\omega},s^{\omega}r_{k}s^{\omega})$ is a \Cs-pair as well. Thus, it follows from~\eqref{eq:trans:caraint} that,
\[
(s^{\omega}r_{k-1}s^{\omega})^{\omega+1} = (s^{\omega}r_{k-1}s^{\omega})^{\omega}s^{\omega}r_{k}s^{\omega}(s^{\omega}r_{k-1}s^{\omega})^{\omega}
\] 
Since $s^{\omega+1} = s^{\omega}r_{k-1}s^{\omega}$ and $s^{\omega} = (s^{\omega}r_{k-1}s^{\omega})^\omega$, this yields,
\[
s^{\omega+1} = (s^{\omega+1})^{\omega+1} = s^{\omega}s^{\omega}r_{k}s^{\omega}s^{\omega} = s^{\omega}r_{k}s^{\omega}
\]
This concludes the proof.
\end{proof}

\bibliographystyle{plain}
\bibliography{main}

\end{document}